\documentclass{article}
\usepackage{latexsym, amssymb,graphicx, amsmath, subfigure,diagbox,threeparttable,amsthm,multicol,appendix}
\usepackage{flushend}
\newtheorem{Theorem}{Theorem}
\newtheorem{Definition}{Definition}
\newtheorem{Lemma}{Lemma}

\newtheorem{Remark}{Remark}
\newtheorem{Example}{Example}

\usepackage{algorithm}
\makeatother
\usepackage{spconf}
\usepackage{indentfirst}

\title{Deterministic Constructions of Binary Measurement Matrices with Various Sizes}
\name{Xin-Ji~Liu,~Shu-Tao~Xia,~Tao Dai\thanks{This research is supported in part by the 973 Program of China (No. 2012CB315803), the National Natural
Science Foundation of China (Nos. 61371078, 61375054), and the Research Fund for the Doctoral Program of Higher Education of China (No. 20130002110051). Corresponding author: Shu-Tao Xia, xiast@sz.tsinghua.edu.cn.}}
\address{Graduate School at Shenzhen, Tsinghua University}
\begin{document}
\ninept
\maketitle
\begin{abstract}
We introduce a general framework to deterministically construct binary measurement matrices for compressed sensing.
The proposed matrices are composed of (circulant) permutation submatrix blocks and zero submatrix blocks, thus making their hardware realization convenient and easy.
Firstly, using the famous Johnson bound for binary constant weight codes, we derive a new lower bound for the coherence of binary matrices with uniform column weights.
Afterwards, a large class of binary \emph{base matrices} with coherence asymptotically achieving this new bound are presented.
Finally, by choosing proper rows and columns from these base matrices, we construct the desired measurement matrices with various sizes and they show empirically comparable performance to that of the corresponding Gaussian matrices.
\end{abstract}
\section{Introduction}
Compressed sensing  (CS) \cite{candes2006a,donoho2006} is a novel sampling technique that samples sparse signals at a rate far lower than the Nyquist-Shannon rate.
Consider a \emph{$k$-sparse} signal $\textit{\textbf{x}}\in\mathbb{R}^{n}$ with at most $k$ nonzero entries,
if we make a linear sampling $\textit{\textbf{y}}=A\textbf{\textit{x}}$ of $\textit{\textbf{x}}$ with the \emph{measurement matrix} $A\in \mathbb{R}^{m\times n}$, where $m<n$,
then $\textit{\textbf{x}}$ could be recovered by solving an \emph{$\ell_1$-minimization} problem \cite{Candes2005} or by a greedy algorithm such as \emph{orthogonal matching pursuit} (OMP) \cite{Tropp2007}.
Actually, if $A$ satisfies the \emph{restricted isometry property} (RIP) \cite{Candes2005} of order $k$ with enough small $0<\delta_k^A<1$, signals with sparsity $O(k)$ can be exactly recovered by $\ell_1$-minimization or OMP \cite[pp. 26]{Foucart2013}, where $\delta_k^A$ denotes the \emph{restricted isometry constant} of $A$.

Many random matrices, such as the Gaussian matrices, have been proved to satisfy RIP of order $k$ \emph{with high probability} if $k\leq O(m/\log(n/k))$ \cite{Baraniuk2008}.
However, there is no guarantee that a specific realization of a random matrix works and some random matrices require lots of storage space.
In contrast, a deterministic matrix is often generated on the fly and RIP could be verified definitely.
Therefore, deterministic measurement matrices are often preferable in practice.

The \emph{coherence} $\mu(A)$ of a deterministic matrix $A$ is often exploited to prove RIP since $\delta_{k}^A \leq (k-1)\mu(A)$ \cite{Bourgain2011},
where
\begin{equation}\label{eq:cohdef}
  \mu(A)\triangleq \max_{1\leq i\neq j\leq n}{\frac{|\langle\textit{\textbf{a}}_{i}, \textit{\textbf{a}}_{j}\rangle|}{||\textit{\textbf{a}}_{i}||_{2}||\textit{\textbf{a}}_{j}||_{2}}},
\end{equation}
$\textit{\textbf{a}}_{1},
\textit{\textbf{a}}_{2}, \ldots, \textit{\textbf{a}}_{n}$ are the $n$ columns of $A$,
$\langle \textit{\textbf{a}}_i, \textit{\textbf{a}}_j\rangle \triangleq \textit{\textbf{a}}_i^T \textit{\textbf{a}}_j$ and for any $\textit{\textbf{z}}=(z_{1}, z_{2}, \ldots, z_{m})^T\in\mathbb{R}^{m}$ , $||\textit{\textbf{z}}||_{2}\triangleq \sqrt{\sum_{i=1}^m z_{i}^2}$.
Therefore, given $\mu(A)$, $A$ satisfies RIP of order
\begin{equation}\label{eq:riporder}
  k <1+\frac{1}{\mu(A)}.
\end{equation}

Recently, binary deterministic matrices have been introduced into compressed sensing due to their simplicity \cite{DeVore2007,Amini2011,Lu2012,Dimakis2012,Tehrani2013,Li2014a}.
For example, let $q$ be a prime power, DeVore proposed a class of binary (before column normalization) $q^2\times q^{r+1}$ matrices satisfying RIP of order $k<q/r+1$, where $1<r<q$ is a constant integer \cite{DeVore2007}.
By using the codewords of orthogonal optical codes as the columns of matrices, Amini \emph{et al.} constructed a class of binary measurement matrices \cite{Amini2011}.
In \cite{Li2014a}, the incidence matrices of several packing designs based on finite geometry are applied into compressed sensing.
These matrices have relatively low coherence and show empirically good performance in compressed sensing.

However, many of them are often based on Galois fields (GF), thus having restrictions to the numbers of rows\footnote{Generally, removing some columns from a matrix will not deteriorate its theoretical (such as coherence and RIP) and empirical performance.}.
Recently, utilizing the parallel structure of Euclidean geometry, we proposed a class of binary measurement matrices with a bit more flexible sizes \cite{sxxl2012}.
In this paper, we introduce more such matrices.
In particular, we focus on the binary matrix $H$ with a constant column weight.
By viewing the columns of $H$ as codewords of a \emph{constant weight code} \cite[pp. 523--531]{MacWilliams1979}, we derive a new lower bound for its coherence $\mu(H)$ with the help of the famous Johnson bound \cite{Johnson1962}, which improves the traditional Welch bound \cite{Welch1974}.
Then we present a subclass of binary (often quasi-cyclic) matrices asymptotically achieving this new bound and some examples from structural low-density parity-check (LDPC) codes \cite{Gallager1962} are given.
Based on these matrices, a general framework is proposed to obtain practical measurement matrices with various sizes.
Finally, simulations show that the proposed matrices perform comparably to, sometimes even better than, the corresponding Gaussian matrices.

\section{Main Results}
\subsection{Coherence of Binary Matrices}
In this part, we analyze the coherence of binary matrices which have uniform column weights $\gamma>1$.

Firstly, some preliminaries are presented.
For any matrix $H\in\{0,1\}^{m\times n}$, there is a \emph{Tanner graph} $G_{H}$ \cite{Tanner1981} corresponding to $H$.
$G_H$ is a bipartite graph comprised of $n$ variable nodes labelled by the elements of $I=\{1,2,\ldots,n\}$, $m$ check nodes labelled by the elements of $J=\{1,2,\ldots,m\}$, and the edge set $E\subseteq\{(i,j):i\in I, j\in J\}$, where there is an edge $(i,j)\in E$ if and only if $h_{ji} = 1$.
The \emph{girth} $g(H)$ of $H$ or $G_H$ is defined as the minimum length of cycles in $G_{H}$.
Girth is always an even number not smaller than 4.
$H$ is said to be $(\gamma,\rho)$-\emph{regular} if $H$ has uniform column weight $\gamma$ and uniform row weight $\rho$.

A binary matrix $H$ with uniform column weight $\gamma$ can be viewed as a collection of codewords (as columns of $H$) of certain binary constant weight codes.
An $(m, d, \gamma)$ constant weight code $\mathcal{C}$ is a set of binary vectors of length $m$, weight $\gamma$ and minimum distance $d$, where $d$ is always an even number.
Let $A(m,d,\gamma)$ be the largest number of codewords in any $(m,d,\gamma)$ constant weight codes, $A(m,d,\gamma)$ could be bounded by the famous Johnson bound \cite{Johnson1962}:
  \begin{equation}\label{eq:johnson2}
    A(m,2\delta,\gamma)\le \lfloor\frac{m}{\gamma}\lfloor\frac{m-1}{\gamma-1}\cdots\lfloor\frac{m-\gamma+\delta}{\delta}\rfloor\cdots\rfloor\rfloor,
  \end{equation}
where $\lfloor x\rfloor$ denotes the largest integer no larger than $x$.

Traditionally, the coherence of a matrix is bounded by the Welch bound \cite{Welch1974}:
\begin{equation}\label{eq:welch}
\mu(A)\geq\sqrt{\frac{n-m}{m(n-1)}}.
\end{equation}
The equality in (\ref{eq:welch}) achieves \emph{if and only if} $A$ is an \emph{equiangular tight frame} (ETF), i.e., A should satisfy the following 3 conditions: (a) the columns of $A$ have unit norm, (b) the rows of $A$ are orthogonal with equal norm, and (c) the inner products between any two different columns of $A$ are equal in modulus \cite{Bandeira2013}.
Therefore, for any binary matrix $H$ with uniform column weight $\gamma>1$, the rows of $H$ will not be orthogonal, thus the Welch bound (\ref{eq:welch}) could not be achieved.

In the following, we analyze the coherence of binary matrices by the Johnson bound.
Consider the binary $m\times n$ matrix $H$ with uniform column weight $\gamma>0$, suppose the maximum inner product of any two columns of $H$ is $\lambda>0$, then $H$ has coherence $\mu(H)=\frac{\lambda}{\gamma}$.
In particular, when $H$ has girth $g(H)>4$, any two distinct columns of $H$ have at most one pair of common `1' at the same row, i.e., $\lambda=1$, we have
\begin{equation}\label{eq:cohbing6}
  \mu(H)=\frac{1}{\gamma}.
\end{equation}

By viewing the column vectors of $H$ as the codewords of an $(m,d,\gamma)$ constant weight code $\mathcal{C}$, then $d=2\gamma-2\lambda$.
From the Johnson bound (\ref{eq:johnson2}), we have the following fact.
\begin{Lemma}\label{lem:johnsonanyg}
  For any binary matrix $H\in\{0,1\}^{m\times n}$ with uniform column weight $\gamma>1$, maximum inner product $0<\lambda<\gamma$ of any two distinct columns, $m$, $n$, $\gamma$ and $\lambda$ should satisfy:
  \begin{equation}\label{eq:mngljohnson2}
    n\le \lfloor\frac{m}{\gamma}\lfloor\frac{m-1}{\gamma-1}\cdots\lfloor\frac{m-\lambda}{\gamma-\lambda}\rfloor\cdots\rfloor\rfloor.
  \end{equation}
\end{Lemma}
In particular, when $H$ has girth $g(H)>4$, we can obtain an explicit lower bound for the coherence of $H$.
\begin{Theorem}\label{th:john2}
  Let $H\in\{0,1\}^{m\times n}$ be a binary matrix with uniform column weight $\gamma>1$, girth $g(H)>4$ and coherence $\mu(H)\ne0$, then
  \begin{equation}\label{eq:cohg6}
    \mu(H)\ge\frac{2n}{n+\sqrt{n^2+4mn(m-1)}}.
  \end{equation}
\end{Theorem}
\begin{proof}
  When $g(H)>4$ and $\mu(H)\ne 0$, $\lambda=1$.
  By (\ref{eq:mngljohnson2}), $n \le \lfloor\frac{m}{\gamma}\lfloor\frac{m-1}{\gamma-1}\rfloor\rfloor\le \frac{m(m-1)}{\gamma(\gamma-1)}$,
 (\ref{eq:cohg6}) follows since $\mu(H)=\frac{1}{\gamma}$.
\end{proof}

\begin{Remark}
By a simple deduction, it is easy to see that (\ref{eq:cohg6}) is always tighter than the Welch bound (\ref{eq:welch}) if $m<n$.
In addition, throughout this paper, we call the binary matrix with uniform column weight $\gamma>0$, girth $g>4$ and coherence $\mu\ne 0$ \emph{(asymptotically) optimal} if the coherence of this matrix (asymptotically) achieves the lower bound (\ref{eq:cohg6}).
\end{Remark}

\begin{Remark}
 Similar to the Johnson bound, (\ref{eq:cohg6}) could be achieved.
 For example, let $H$ be the point-line incidence matrix (rows of $H$ corresponding to the points and columns to the lines) of the Euclidean plane $EG(2,q)$, where $q$ is a prime power. $H$ is a $(q,q+1)$-regular matrix with the size $q^2\times (q^2+q)$, $g(H)=6$, and it is easy to verify that (\ref{eq:cohg6}) is achieved, see \cite{sxxl2012} for more details of $H$ and its application to compressed sensing.
\end{Remark}

\subsection{A Subclass of Asymptotically Optimal Binary Matrices in Terms of Coherence}
In this part, we show a subclass of binary matrices with coherence asymptotically achieving the lower bound (\ref{eq:cohg6}).
Later on, they will be used to obtain the desired measurement matrices with various sizes and empirically good performance.

Consider an ${s^2\times s^2}$ \emph{base matrix} as follows
\begin{equation}\label{eq:basem}
  H = [H_{i,j}], \quad 1\leq i,j\leq s,
\end{equation}
where $s>1$ and $H_{i,j}\in\{0,1\}^{s\times s}$ is either a permutation block or a zero block $\textbf{0}=\{0\}^{s\times s}$.
A \emph{permutation block} $B\in\{0,1\}^{s\times s}$ is a square matrix with each row and each column having exactly one element `1'.
If $B$ is also cyclic, then $B$ is called a \emph{circulant permutation block}.
Each $[H_{i,1}, H_{i,2},\ldots,H_{i,s}]$ (or $[H_{1,j}^T, H_{2,j}^T,\ldots,H_{s,j}^T]^T$) of $H$ is called a \emph{row-block} (or \emph{column-block}) of $H$.
$H$ satisfies the following two properties.
\begin{itemize}
  \item (P1) Every column-block of $H$ has exactly $t$ zero blocks, so does each row-block , i.e., $H$ is $(s-t,s-t)$-regular, where $0\le t\ll s$ is a small constant.
  \item (P2) The girth of $H$ is larger than 4, i.e., $g(H)>4$.
\end{itemize}

\begin{Remark}
  The $s^2\times s^2$ base matrix $H$ has coherence $\mu(H)=\frac{1}{s-t}$.
  According to Theorem \ref{th:john2}, the (nonzero) coherence of any $s^2\times s^2$ binary matrix with uniform column weight $\gamma>1$ and girth larger than 4 has the lower bound $\frac{2}{1+\sqrt{4s^2-3}}\rightarrow\frac{1}{s}$ if $s\rightarrow\infty$.
  Since $t\ge 0$ is a small constant, the coherence of the base matrix $H$ is asymptotically optimal.
\end{Remark}

In addition, for some submatrices of the base matrices, their coherences are also asymptotically optimal.
Let $A(\gamma,s,t)$ be a $\gamma s\times s^2$ submatrix of the base matrix $H$ by simply choosing the first $\gamma$ row-blocks of $H$, i.e.,
\begin{equation}\label{eq:gamm}
  A(\gamma,s,t)\triangleq [H_{i,j}], \;\;1\leq i\leq \gamma\le s,\;\;1\leq j\leq s.
\end{equation}

\begin{Remark}\label{rem:asymopt}
  Suppose $\gamma=cs$, where $0<c<1$ is a constant such that $cs$ is an integer.
  When $t=0$, $A(cs,s,0)$ is a $(cs,s)$-regular matrix with coherence $\mu(A(cs,s,0))=\frac{1}{cs}$.
  According to (\ref{eq:cohg6}), for any binary $cs^2\times s^2$ matrix with uniform column weight, girth larger than 4 and nonzero coherence, its coherence has the lower bound $\frac{1}{0.5+\sqrt{c^2s^2+0.25-c}}\rightarrow\frac{1}{cs}$
  if $s\rightarrow\infty$.
  Therefore, the submatrix $A(cs,s,0)$ is also asymptotically optimal in terms of coherence.
\end{Remark}

In the following, we review several examples of satisfactory base matrices from structured (often quasi-cyclic) LDPC codes.

\begin{Example}[\cite{Fossorier2004, Liu2013}]\label{con:add}
    Let $H=H(q,q)$, where $q$ is an odd prime and $H(q,q)$ is the binary matrix defined in (7) in \cite{Liu2013} with $r=q$.
    Then $H\in\{0,1\}^{q^2\times q^2}$ is a $(q,q)$-regular base matrix with $t=0$.
\end{Example}

\begin{Example}[\cite{Ge2006,Li2012a}]\label{con:bjq}
Let $H = H^{(1)}(q,q,0)$, where $q$ is a prime power and $H^{(1)}(q,q,0)$ is the parity-check matrix of a first class of B--J based LDPC code proposed in \cite[Section III.A]{Ge2006}.
Then $H\in\{0,1\}^{q^2\times q^2}$ is a $(q,q)$-regular base matrix with $t=0$.
\end{Example}


Let $GF(q) = \{\alpha^{-\infty}= 0, \alpha^0 = 1, \alpha, \ldots, \alpha^{q-2}\}$ be a Galois field with primitive element $\alpha$.
Establish a one-to-one \emph{$(q-1)$-fold correspondence} between the elements in $GF(q)$ and the matrices $P\in\{0,1\}^{(q-1)\times (q-1)}$ as follows:
\begin{itemize}
  \item $0$ is mapped to the zero block $\textbf{0}=\{0\}^{(q-1)\times (q-1)}$;
  \item $\alpha^i$ is mapped to a circulant permutation block $P_{q-1}^i$, where $0\leq i\leq q-2$,
  \begin{eqnarray}\label{eq:mpq}
        P_{q-1}\triangleq\left[
        \begin{array}{ccccc}
            0&1&0&\cdots&0\\
            0&0&1&\cdots&0\\
        \vdots&\vdots&\vdots&\ddots&\vdots\\
        0&0&0&\cdots&1\\
        1&0&0&\cdots&0\\
    \end{array}
    \right]_{(q-1)\times (q-1)},
\end{eqnarray}
$P_{q-1}^i$ denotes the $i$-th power of $P_{q-1}$ and $P_{q-1}^0\triangleq I_{q-1}$ is the identity matrix of order $q-1$.
\end{itemize}

In the following Examples \ref{con:rslatin} and \ref{con:latin}, we obtain the base matrix $H$ by firstly constructing a matrix $L\in GF(q)^{(q-1)\times(q-1)}$ based on the \emph{Latin square} and then replacing each element in $L$ with a circulant permutation block or a zero block $P\in\{0,1\}^{(q-1)\times (q-1)}$.

\begin{Definition}
  An Latin square of order $n$ is an $n\times n$ matrix with $n$ distinct symbols, each of which occurrs exactly once in each row and exactly once in each column.
\end{Definition}

\begin{Example}[\cite{Lan2007,Zeng2008,Zhang2010a}]\label{con:rslatin}
    Let $\beta$ be a nonzero element in $GF(q)$ and $L_{\rm RS}(\beta)$ be the following Reed-Solomon codes based (cyclic) Latin square of order $q-1$ over $GF(q)\setminus\{-\beta\}$:
\begin{eqnarray*}
L_{\rm RS}(\beta) =
\left[
\begin{array}{cccc}
1-\beta&\alpha-\beta&\ldots&\alpha^{q-2}-\beta\\
\alpha^{q-2}-\beta&1-\beta&\ldots&\alpha^{q-3}-\beta\\
\vdots&\vdots&\ddots&\vdots\\
\alpha-\beta&\alpha^{2}-\beta&\ldots&1-\beta
\end{array}
\right].
\end{eqnarray*}
Expand $L_{\rm RS}(\beta)$ by replacing each entry with a circulant permutation block or a zero block according to the $(q-1)$-fold correspondence, and then we could get a quasi-cyclic base matrix $H\in\{0,1\}^{(q-1)^2\times (q-1)^2}$ with $t=1$.
Note that no matter which nonzero $\beta$ is chosen, there is exactly one $0$ in each row and exactly one $0$ in each column of $L_{\rm RS}(\beta)$.
Therefore, the resulting $H$ is a $(q-2,q-2)$-regular matrix.
Finally, as there are $q-1$ nonzero elements $\beta\in GF(q)$, there will be $q-1$ such Latin squares $L_{\rm RS}(\beta)$ and thus $q-1$ such base matrices $H$.
\end{Example}
\begin{Example}[\cite{Zhang2010a}]\label{con:latin}
Let $\beta$ be any nonzero element in $GF(q)$ and $\bar L(\beta)=W$ be the Latin square of order $q$ over $GF(q)$ in \cite[Equation (11)]{Zhang2010a} with $\beta=\eta$.
Choose the following $(q-1)\times(q-1)$ submatrix $L(\beta)$ of $\bar L(\beta)$, $L(\beta)=$
\begin{eqnarray*}
\left[
\begin{array}{cccc}
\beta-1&\beta-\alpha&\ldots&\beta-\alpha^{q-2}\\
\alpha\beta-1&\alpha\beta-\alpha&\ldots&\alpha\beta-\alpha^{q-2}\\
\vdots&\vdots&\ddots&\vdots\\
\alpha^{q-2}\beta-1&\alpha^{q-2}\beta-\alpha&\ldots&\alpha^{q-2}\beta-\alpha^{q-2}
\end{array}
\right].
\end{eqnarray*}
Expand $L(\beta)$ according to the $(q-1)$-fold correspondence.
In this way, we obtain a quasi-cyclic and $(q-2,q-2)$-regular base matrix $H\in\{0,1\}^{(q-1)^2\times (q-1)^2}$ with $t=1$.
\end{Example}

\begin{Example}[\cite{Ge2006,Li2012a}]\label{con:bjq1}
  Let $q$ be a prime power. Let $H = H^{(2)}(q-1,q-1,0)$, where $H^{(2)}(q-1,q-1,0)$ is the parity-check matrix of a second class of B--J based LDPC code proposed in \cite[Section III.B]{Ge2006}.
  Then $H\in\{0,1\}^{(q-1)^2\times (q-1)^2}$ is a quasi-cyclic and $(q-2,q-2)$-regular base matrix with $t=1$.
\end{Example}



\subsection{General Framework of Matrix Constructions}
In this part, we give the general framework to deterministically construct binary measurement matrices, see Algorithm \ref{alg:construction}.
Note that in the second step of Algorithm \ref{alg:construction}, we choose the $s^2\times s^2$ base matrix $H$ in such a way as to make the resulting $A$ have the smallest coherence.
In practice, we often require $m\gg\sqrt{n}$, thus the outputted matrix will have small coherence and empirically good performance.
For example, $m$ scales linearly with $n$, i.e., $m=cn$, where $0<c<1$ is a constant.
See the following Theorem \ref{th:coh} for a formalized explanation.
\begin{algorithm}[htbp]
\textbf{Input}: Matrix size $m$ and $n$.\\
\textbf{Output}: A binary measurement matrix $A\in \{0,1\}^{m\times n}$.\\
\textbf{Steps}:

(1) Base matrix construction: construct several classes of $\bar s^2\times \bar s^2$ base matrices satisfying (P1) and (P2) with $t\ll \bar{s}$.

(2) Base matrix selection: choose an $s^2\times s^2$ matrix $H$ among these base matrices such that $s\ge \sqrt{n}$ and $(m/s-t)$ is as large as possible.

(3) Extra elements deletion: remove the last $s^2-m$ rows and the last $s^2-n$ columns of $H$ and output the resulting submatrix as $A$.

\caption{Deterministic Construction of Binary Measurement Matrices with Various Sizes} \label{alg:construction}
\end{algorithm}

\begin{Theorem}\label{th:coh}
  For any measurement matrix $A\in \{0,1\}^{m\times n}$ constructed by Algorithm \ref{alg:construction}, we have
  \begin{equation}\label{eq:coharbm}
    \mu(A)\leq \frac{1}{\gamma-t},
  \end{equation}
  where $\gamma = \left\lfloor\frac{m}{s}\right\rfloor$,
$0\leq t\ll s$ is a fixed integer.
\end{Theorem}
\begin{proof}
    According to (P1), the minimum possible column weight of $A$ is $\gamma-t$.
    From (P2), the inner product of any two columns of $A$ is at most 1.
    By (\ref{eq:cohdef}), (\ref{eq:coharbm}) follows directly.
\end{proof}
\begin{Remark}
  As stated in Remark \ref{rem:asymopt}, when $\gamma=\frac{m}{s}=cs$, $n=s^2$ and $t=0$, the binary matrix $A(cs,s,0)$ outputted by Algorithm \ref{alg:construction} is asymptotically optimal in terms of coherence.
  In other cases, the structure (and thus the coherence) of the resulting matrix $A$ with $m=cn$ and $n=s^2$ is very close to that of $A(cs,s,0)$.
  Moreover, removing columns of a measurement matrix will not deteriorate its empirical performance.
  Therefore, it is reasonable to conjecture that the measurement matrices obtained by Algorithm \ref{alg:construction} will often perform well in practice and this will be verified by the following experimental results.
\end{Remark}

\section{Experimental Results}\label{sec:simu}
In the following simulations, for each measurement matrix $A$ and each $k$-sparse signal $\textit{\textbf{x}}$, we conduct an experiment using $M=1000$ Monte Carlo trials.
In the $i$-th trial, a relative recovery error $e_i = ||\textit{\textbf{x}}^*-\textit{\textbf{x}}||_2/||\textit{\textbf{x}}||_2$ is computed, where $\textit{\textbf{x}}^*$ denotes the recovered signal.
If $e_i\leq 0.001$, we declare this recovery to be ``perfect''.
Finally, an average percentage of perfect recovery over the $M$ trials is obtained and shown as a point in the figures.

At first, we give an example to show the empirical effectiveness for the base matrix selecting strategy in the second step of Algorithm \ref{alg:construction}.
Suppose only one class of base matrices are constructed in the first step, such as the base matrices in Example \ref{con:rslatin}, and now we want to construct a $100\times 300$ binary measurement matrix.
Since $\sqrt{300}=17.32$, we can set $q$ to be $19, 23$, or even larger prime power.
The OMP recovery performance of the desired measurement matrices obtained by setting $q=19$, $q=23$ and the Gaussian matrix (`Rnd') with the same size are shown in Fig. \ref{fig:schoose}.
\begin{figure}[htbp]
   \centering
   \includegraphics[width=0.45\textwidth]{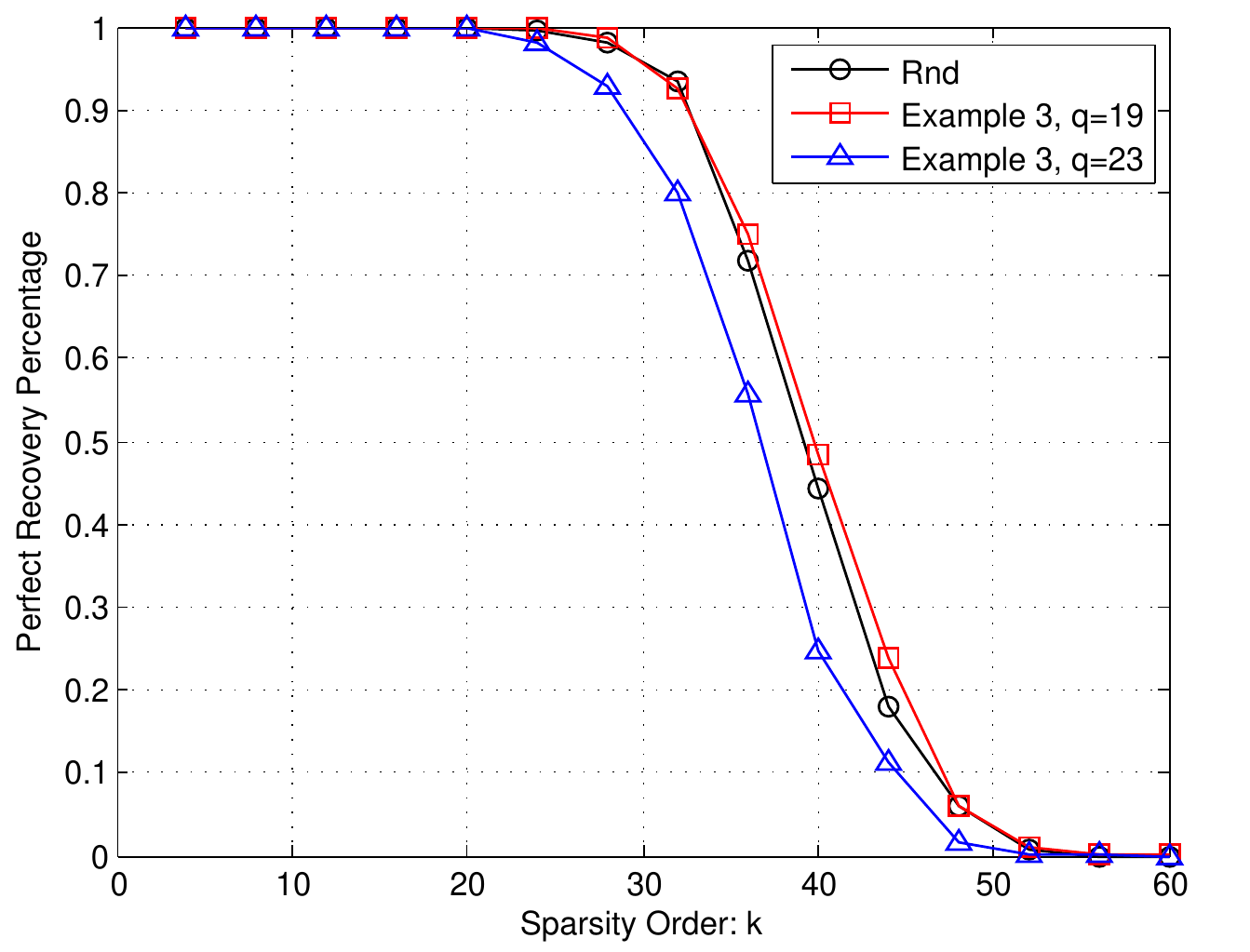}
   \caption{Empirical performance of the $100\times 300$ measurement matrices obtained by Example \ref{con:rslatin} by setting $q=19$ and $q=23$ and the corresponding Gaussian random matrix under OMP recovery.}
   \label{fig:schoose}
\end{figure}
It is clear that the matrix based on Example \ref{con:rslatin} with $q=19$ is better than that with $q=23$, which agrees with the base matrix choosing strategy in the second step of Algorithm \ref{alg:construction}.

In the following, we consider several binary measurement matrices based on the base matrices in Examples \ref{con:add}--\ref{con:bjq1}, see Fig. \ref{fig:comparesmall} for the empirical performance of these matrices with small sizes and Fig. \ref{fig:comparelarge} for that of matrices with larger sizes.

Let $q=31$ in Example \ref{con:add}, $q=32$ in Examples \ref{con:bjq}--\ref{con:bjq1},
$\beta=1$ in Examples \ref{con:rslatin} and \ref{con:latin}.
For each Example \ref{con:add}--\ref{con:bjq1}, construct $3$ measurement matrices with sizes $190\times940$, $225\times950$, and $260\times960$ by removing the last extra rows and columns from the $5$ different base matrices.
See Fig. \ref{fig:comparesmall} for their empirical performance and the corresponding Gaussian matrices.
\begin{figure}[htbp]
   \centering
   \includegraphics[width=0.45\textwidth]{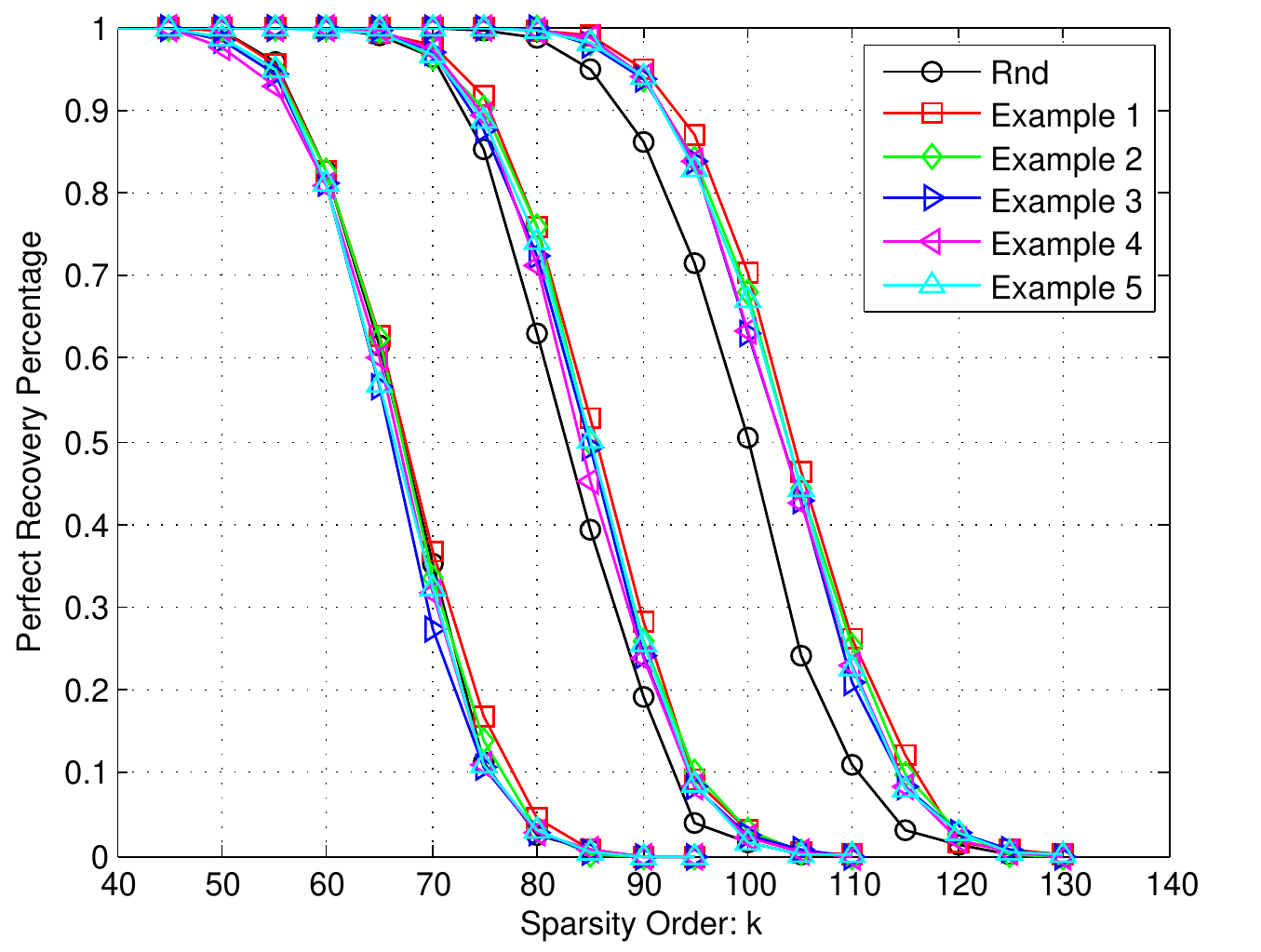}
   \caption{Empirical performance of the proposed measurement matrices and the corresponding Gaussian random matrices with sizes $190\times940$, $225\times950$, and $260\times960$ (the three curve bundles from left to right, respectively) under OMP recovery.}
   \label{fig:comparesmall}
\end{figure}

Let $q=61$ in Example \ref{con:add} and $q=64$ in Example \ref{con:bjq}--\ref{con:bjq1}.
For each Example \ref{con:add}--\ref{con:bjq1}, construct $3$ matrices with sizes $450\times3500$, $500\times3600$, and $550\times3700$.
See Fig. \ref{fig:comparelarge} for their empirical performance.
\begin{figure}[htbp]
   \centering
   \includegraphics[width=0.45\textwidth]{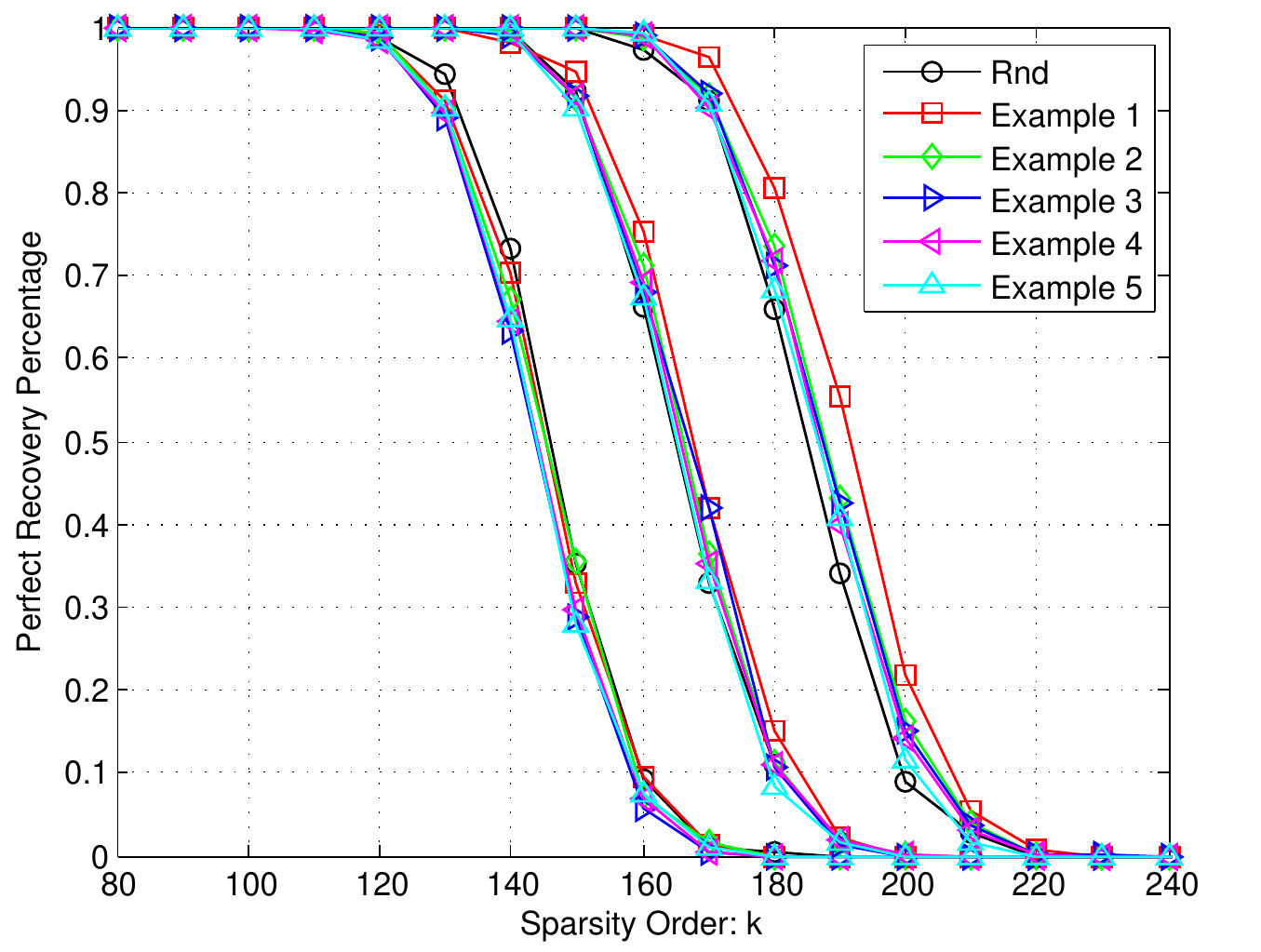}
   \caption{Empirical performance of the proposed measurement matrices and the corresponding Gaussian matrices with sizes $450\times3500$, $500\times3600$, and $550\times3700$ (the three curve bundles from left to right, respectively) under OMP recovery.}
   \label{fig:comparelarge}
\end{figure}

In Figs. \ref{fig:comparesmall} and \ref{fig:comparelarge}, all of the proposed matrices perform as well as, sometimes even better than, the corresponding Gaussian matrices.
In addition, it is easy to see that the matrices from Example \ref{con:add} often perform slightly better than those from Example \ref{con:bjq}--\ref{con:bjq1} due to the specific matrix sizes.
Simple computations on the upper bounds of coherence (according to Theorem \ref{th:coh}) show that each coherence upper bound of the six matrices obtained by Example \ref{con:add} is smaller than (or sometimes equal to) that of other examples.
This also agrees with the base matrix selecting strategy in the second step of Algorithm \ref{alg:construction}.
\section{Conclusions and Discussions}
This paper has introduced a general framework to deterministically construct binary measurement matrices with various sizes and empirically good performance.
In particular, some of them are also shown to be asymptotically optimal according to a new lower bound of coherence derived with the help of the famous Johnson bound.
Moreover, these matrices are binary, sparse, and mostly quasi-cyclic, which will benefit the hardware implementation.

This paper mainly focuses on binary matrices with girth larger than 4.
However, as has been indicated by Lu \cite{Lu2012}, some empirically even better binary matrices lie in the region of girth $g=4$.
In addition, a $(q^2+1)\times q(q^2+1)$ binary measurement matrix with uniform column weight $\gamma=q+1$ and $\lambda=2$ (thus girth $g=4$) has been proposed in \cite{Li2014a}.
It is easy to verify that this matrix achieves the Johnson bound and Equation (\ref{eq:mngljohnson2}) in this paper and they are also shown to perform empirically well in \cite{Li2014a}.
As a result, it will be interesting to carry out some theoretical analysis explicitly on binary matrices with $g=4$ and construct more such (asymptotically) optimal matrices which may show perhaps better performance in practice.
\newpage
\bibliographystyle{IEEEtran}
\bibliography{IEEEabrv,OptBin}
\end{document}